\documentclass{article}

\usepackage{amsmath}
\usepackage{amssymb}
\usepackage{amsthm}
\usepackage{float}
\usepackage{fullpage}

\newtheorem{theorem}{Theorem}
\newtheorem{lemma}[theorem]{Lemma}
\newtheorem{corollary}[theorem]{Corollary}
\newtheorem{claim}[theorem]{Claim}

\theoremstyle{definition}
\newtheorem{definition}{Definition}

\floatstyle{ruled}
\newfloat{algorithm}{htbp}{loa}
\floatname{algorithm}{Algorithm}

\newcommand{\Supp}{{\mathop{\rm supp}\nolimits}}
\newcommand{\Vol}{{\mathop{\rm Vol}\nolimits}}
\newcommand{\Poly}{{\mathop{\rm poly}\nolimits}}

\def\eps{\epsilon}

\begin{document}

\title{Random Tensors and Planted Cliques}
\author{S. Charles Brubaker \;\;\;\;\;\;  Santosh S. Vempala\\
Georgia Institute of Technology\\
Atlanta, GA 30332 \\
\{brubaker,vempala\}@cc.gatech.edu}
\date{}

\maketitle
\begin{abstract}
The $r$-parity tensor of a graph is a generalization of the adjacency
matrix, where the tensor's entries denote the parity of the number of
edges in subgraphs induced by $r$ distinct vertices. For $r=2$, it is
the adjacency matrix with $1$'s for edges and $-1$'s for nonedges. It
is well-known that the $2$-norm of the adjacency matrix of a random
graph is $O(\sqrt{n})$. Here we show that the $2$-norm of the
$r$-parity tensor is at most $f(r)\sqrt{n}\log^{O(r)}n$, answering a
question of Frieze and Kannan \cite{FK08} who proved this for
$r=3$. As a consequence, we get a tight connection between the planted
clique problem and the problem of finding a vector that approximates
the $2$-norm of the $r$-parity tensor of a random graph.  Our proof
method is based on an inductive application of concentration of
measure.
\end{abstract}

\section{Introduction}

It is well-known that a random graph $G(n,1/2)$ almost surely has a
clique of size $(2+o(1))\log_2 n$ and a simple greedy algorithm finds
a clique of size $(1+o(1))\log_2 n$. Finding a clique of size even
$(1+\eps)\log_2 n$ for some $\eps > 0$ in a random graph is a
long-standing open problem posed by Karp in 1976 \cite{Karp76} in his
classic paper on probabilistic analysis of algorithms.

In the early nineties, a very interesting variant of this question was
formulated by Jerrum \cite{Jerrum92} and by Kucera
\cite{Kucera95}. Suppose that a clique of size $p$ is planted in a
random graph, i.e., a random graph is chosen and all the edges within
a subset of $p$ vertices are added to it. Then for what value of $p$
can the planted clique be found efficiently? It is not hard to see
that $p > c\sqrt{n\log n}$ suffices since then the vertices of the
clique will have larger degrees than the rest of the graph, with high
probability \cite{Kucera95}. This was improved by Alon et al
\cite{AKS98} to $p=\Omega(\sqrt{n})$ using a spectral approach. This
was refined by McSherry \cite{McSherry01} and considered by Feige and
Krauthgamer in the more general semi-random model \cite{FeigeK00}.
For $p \geq 10 \sqrt{n}$, the following simple algorithm works: form a
matrix with $1$'s for edges and $-1$'s for nonedges; find the largest
eigenvector of this matrix and read off the top $p$ entries in
magnitude; return the set of vertices that have degree at least $3p/4$
within this subset.

The reason this works is the following: the top eigenvector of a symmetric matrix $A$
can be written as
\[
\max_{x: \|x\|=1} x^TAx = \max_{x: \|x\| = 1} \sum_{ij} A_{ij}x_i x_j
\]
maximizing a quadratic polynomial over the unit sphere. The maximum
value is the spectral norm or $2$-norm of the matrix. For a random
matrix with $1,-1$ entries, the spectral norm (largest eigenvalue) is
$O(\sqrt{n})$. In fact, as shown by F\"{u}redi and Koml\'{o}s \cite{FK81,
Vu05}, a random matrix with i.i.d. entries of variance at most $1$ has
the same bound on the spectral norm. On the other hand, after planting
a clique of size $\sqrt{n}$ times a sufficient constant factor, the
indicator vector of the clique (normalized) achieves a higher
norm. Thus the top eigenvector points in the direction of the clique
(or very close to it).

Given the numerous applications of eigenvectors (principal
components), a well-motivated and natural generalization of this
optimization problem to an $r$-dimensional tensor is the following:
given a symmetric tensor $A$ with entries $A_{k_1k_2\ldots k_r}$, find
\[
\|A\|_2 = \max_{x: \|x\|=1} A(x,\ldots,x),
\]
where
\[
A(x^{(1)},\ldots,x^{(r)}) = \sum_{i_1i_2\ldots i_r} A_{i_1i_2\ldots
i_r}x^{(1)}_{i_1}x^{(2)}_{i_2}\ldots x^{(r)}_{i_r}.
\]
The maximum value is the spectral norm or $2$-norm of the tensor.  The
complexity of this problem is open for any $r > 2$, assuming the
entries with repeated indices are zeros.

A beautiful application of this problem was given recently by Frieze
and Kannan \cite{FK08}. They defined the following tensor associated
with an undirected graph $G=(V,E)$:
\[
A_{ijk} = E_{ij}E_{jk}E_{ki}
\]
where $E_{ij}$ is $1$ is $ij \in E$ and $-1$ otherwise, i.e.,
$A_{ijk}$ is the parity of the number of edges between $i,j,k$ present
in $G$. They proved that for the random graph $G_{n,1/2}$, the
$2$-norm of the random tensor $A$ is $\tilde{O}(\sqrt{n})$, i.e.,
\[
\sup_{x: \|x\| = 1} \sum_{i,j,k} A_{ijk}x_i x_j x_k \le
C \sqrt{n}\log^c n
\]
where $c,C$ are absolute constants. This implied that if such a
maximizing vector $x$ could be found (or approximated), then we could
find planted cliques of size as small as $n^{1/3}$ times
polylogarithmic factors in polynomial time, improving substantially on
the long-standing threshold of $\Omega(\sqrt{n})$.

Frieze and Kannan ask the natural question of whether this connection
can be further strengthened by going to $r$-dimensional tensors for $r
> 3$. The tensor itself has a nice generalization. For a given graph
$G=(V,E)$ the $r$-parity tensor is defined as follows. Entries with
repeated indices are set to zero; any other entry is the parity of the
number of edges in the subgraph induced by the subset of vertices
corresponding to the entry, i.e.,
\[
A_{k_1,\ldots,k_r} = \prod_{1 \le i < j \le r} E_{k_ik_j}.
\]
Frieze and Kannan's proof for $r=3$ is combinatorial (as is the proof
by F\"{u}redi and Koml\'{o}s for $r=2$), based on counting the number
of subgraphs of a certain type. It is not clear how to extend this
proof.

Here we prove a nearly optimal bound on the spectral norm of this
random tensor for any $r$. This substantially strengthens the
connection between the planted clique problem and the tensor norm
problem. Our proof is based on a concentration of measure approach. In
fact, we first reprove the result for $r=3$ using this approach and
then generalize it to tensors of arbitrary dimension.  We show that
the norm of the subgraph parity tensor of a random graph is at most
$f(r)\tilde{O}(\sqrt{n})$ whp.  More precisely, our main theorem is
the following.
\begin{theorem}\label{thrm:r-tensor-norm}
There is a constant $C_1$ such that with probability at least $1 -
n^{-1}$ the norm of the $r$-dimensional subgraph parity tensor $A:
[n]^r \rightarrow \{-1,1\}$ for the random graph $G_{n,1/2}$ is
bounded by
\[
\|A\|_2 \leq C_1^r r^{(5r-1)/2} \sqrt{n} \log^{(3r-1)/2} n.
\]
\end{theorem}

The main challenge to the proof is the fact that the entries of the
tensor $A$ are not independent.  Bounding the norm of the tensor where
every entry is independently $1$ or $-1$ with probability $1/2$ is
substantially easier via a combination of an $\eps$-net and a
Hoeffding bound. In more detail, we approximate the unit ball with a
finite (exponential) set of vectors.  For each vector $x$ in the
discretization, the Hoeffding inequality gives an exponential tail
bound on $A(x,\ldots,x)$.  A union bound over all points in the
discretization then completes the proof.  For the parity tensor,
however, the Hoeffding bound does not apply as the entries are not
independent. Moreover, all the $\binom{n}{r}$ entries of the tensor
are fixed by just the $\binom{n}{2}$ edges of the graph. In spite of
this heavy inter-dependence, it turns out that $A(x,\ldots,x)$ does
concentrate. Our proof is inductive and bounds the norms of vectors
encountered in a certain decomposition of the tensor polynomial.

Using Theorem \ref{thrm:r-tensor-norm}, we can show that if the norm
problem can be solved for tensors of dimension $r$, one can find
planted cliques of size as low as $Cn^{1/r}\Poly(r,\log n)$.  While
the norm of the parity tensor for a random graph remains bounded, when
a clique of size $p$ is planted, the norm becomes at least $p^{r/2}$
(using the indicator vector of the clique).  Therefore, $p$ only needs
to be a little larger than $n^{1/r}$ in order for the the clique to
become the dominant term in the maximization of $A(x,\ldots,x)$.  More
precisely, we have the following theorem.
\begin{theorem}\label{thrm:reduction}
Let $G$ be random graph $G_{n,1/2}$ with a planted clique of size $p$,
and let $A$ be the $r$-parity tensor for $G$. For $\alpha \leq 1$, let
$T(n,r)$ be the time to compute a vector $x$ such that $A(x,\ldots,x)
\ge \alpha^r\|A\|_2$ whp.  Then, for $p$ such that
\[
n \geq p > C_0 \alpha^{-2} r^5 n^{1/r} \log^3n,
\]
the planted clique can be recovered with high probability in time
$T(n,r)+ \Poly(n)$, where $C_0$ is a fixed constant.
\end{theorem}

On one hand, this highlights the benefits of finding an efficient
(approximation) algorithm for the tensor problem.  On the other, given
the lack of progress on the clique problem, this is perhaps evidence
of the hardness of the tensor maximization problem even for a natural
class of random tensors.  For example, if finding a clique of size
$\tilde{O}(n^{1/2 - \epsilon})$ is hard, then by setting $\alpha =
n^{1/2r + \epsilon/2 - 1/4}$ we see that even a certain polynomial
approximation to the norm of the parity tensor is hard to achieve.
\begin{corollary}
Let $G$ be random graph $G_{n,1/2}$ with a planted clique of size $p$,
and let $A$ be the $r$-parity tensor for $G$.  Let $\epsilon > 0$ be a
small constant and let $T(n,r)$ be the time to compute a vector $x$
such that $A(x,\ldots,x) \ge n^{1/2 + r\epsilon/2 -
r/4}\|A\|_2$. Then, for
\[
p \geq C_0 r^5 n^{\frac{1}{2}-\epsilon} \log^3n,
\]
the planted clique can be recovered with high probability in time
$T(n,r)+ \Poly(n)$, where $C_0$ is a fixed constant.
\end{corollary}

\subsection{Overview of analysis}
The majority of the paper is concerned with proving Theorem
\ref{thrm:r-tensor-norm}.  In Section \ref{sec:discretize}, we first
reduce the problem of bounding $A(\cdot)$ over the unit ball to
bounding it over a discrete set of vectors that have the same value in
every non-zero coordinate.  In Section \ref{sec:off-diag}, we further
reduce the problem to bounding the norm of an off-diagonal block of
$A$, using a method of Frieze and Kannan.  This enables us to assume
that if $(k_1,\ldots,k_r)$ is a valid index, then the random variables
$E_{k_i,k_j}$ used to compute $A_{i_1,\ldots,i_r}$ are independent.
In Section \ref{sec:concentration}, we prove a large deviation
inequality (Lemma \ref{thrm:concentration}) that allows us to bound
norms of vectors encountered in a certain decomposition of the tensor
polynomial.  This inequality gives us a considerably sharper bound
than the Hoeffding or McDiarmid inequalities in our context.  We then
apply this lemma to bound $\|A\|_2$ for $r=3$ as a warm-up and then
give the proof for general $r$ in Section \ref{sec:r-tensor-norm}.

In Section \ref{sec:reduction} we prove Theorem \ref{thrm:reduction}.
We first show that any vector $x$ that comes close to maximizing
$A(\cdot)$ must be close to the indicator vector of the clique
(Lemma \ref{thrm:U-approx}).  Finally, we show that given such a
vector it is possible to recover the clique (Lemma
\ref{thrm:recover}).

\section{Preliminaries}

\subsection{Discretization}\label{sec:discretize}
The analysis of $A(x,\ldots,x)$ is greatly simplified when $x$ is
proportional to some indicator vector.  Fortunately, analyzing these
vectors is sufficient, as any vector can be approximated as a linear
combination of relatively few indicator vectors.

For any vector $x$, we define $x^{(+)}$ to be vector such that
$x^{(+)}_i = x_i$ if $x_i > 0$ and $x^{(+)}_i = 0$ otherwise.
Similarly, let $x^{(-)}_i = x_i$ if $x_i < 0$ and $x^{(-)}_i = 0$
otherwise.  For a set $S \subseteq [n]$, let $\chi^S$ be the indicator
vector for $S$, where the $ith$ entry is $1$ if $i \in S$ and $0$
otherwise.

\begin{definition}[Indicator Decomposition]
For a unit vector $x$, define the sets $S_1,\ldots$ and $T_1,\ldots$
through the recurrences
\[
S_j = \left\{ i \in [n] : (x^{(+)} -  \sum_{k=1}^{j-1} 2^{-k} \chi^{S_k})_i 
> 2^{-j} \right\}.
\]
and 
\[
T_j = \left\{ i \in [n] : (x^{(-)} -  \sum_{k=1}^{j-1} 2^{-k} \chi^{S_k})_i 
< -2^{-j} \right\}.
\]
Let $y_0(x) = 0$. For $j \geq 1$, let $y^{(j)}(x) = 2^{-j} \chi^{S_j}$ and let
$y^{(-j)}(x) = -2^{-j} \chi^{T_j}$.  We call the set
$\{y^{(j)}(x)\}_{-\infty}^\infty$ the \emph{indicator decomposition of $x$}.
\end{definition}

Clearly,
\[
\|y^{(i)}(x)\| \leq \max \{\|x^{(+)}\|,\|x^{(-)}\|\} \leq 1.
\]
and 
\[
\left\|x - \sum_{j=-N}^N y^{(j)}(x) \right\| \leq \sqrt{n} 2^{-N}.
\]
We use this decomposition to prove the following theorem.

\begin{lemma}\label{thrm:U-approx}
Let
\[
U = \{ k |S|^{-1/2} \chi^S : S \subseteq [n], k \in \{-1,1\}\}.
\]
For any tensor $A$ over $[n]^r$ where $\|A\|_\infty \leq 1$
\[
\max_{x^{(1)},\ldots,x^{(r)} \in B(0,1)} A(x^{(1)},\ldots x^{(r)})
\leq (2 \lceil r \log n \rceil )^r \max_{x^{(1)}, \ldots, x^{(r)} \in U}
A(x^{(1)},\ldots,x^{(r)})
\]
\end{lemma}
\begin{proof}
Consider a fixed set of vectors $x^{(1)},\ldots,x^{(r)}$ and let $N =
\lceil r \log_2 n \rceil$.  For each $i$, let
\[
\hat{x^{(i)}} = \sum_{j=-N}^N y^{(j)}(x^{(i)}).
\]

We first show that replacing $x^{(i)}$ with $\hat{x^{(i)}}$ gives a good
approximation to $A(x^{(1)},\ldots, x^{(r)})$.  Letting $\epsilon$ be
the maximum difference between an $x^{(i)}$ and its approximation, we
have that
\[
\max_{i \in [r]} \|x^{(i)} - \hat{x^{(i)}}\| = \epsilon \leq \frac{n^{r/2}}{2r}
\]
Because of the multilinear form of $A(\cdot)$ we have
\begin{eqnarray*}
| A(x^{(1)},\ldots,x^{(r)}) - A(\hat{x^{(1)}},\ldots,\hat{x^{(r)})}| & \leq &
\sum_{i=1}^r \epsilon^i r^i \|A\|\\
& \leq & \frac{\epsilon r}{1 - \epsilon r} \|A\|\\
& \leq & n^{-r/2} \|A\|\\
& \leq & 1.
\end{eqnarray*}

Next, we bound $A(\hat{x^{(1)}},\ldots, \hat{x^{(r)})}$. For
convenience, let $Y^{(i)} = \cup_{j=-N}^N y^{(j)}(x^{(i)})$.  Then using the
multlinear form of $A(\cdot)$ and bounding the sum by its maximum term, we have
\begin{eqnarray*}
A(\hat{x}^{(1)},\ldots, \hat{x}^{(r)}) 
& \leq & (2N)^r
\max_{v^{(1)} \in Y^{(1)},\ldots,v^{(r)} \in Y^{(r)}} 
A(v^{(1)},\ldots, v^{(r)}) \\
& \leq & (2N)^r
\max_{v^{(1)},\ldots,v^{(r)} \in U} A(v^{(1)},\ldots, v^{(r)}).
\end{eqnarray*}
\end{proof}

\subsection{Sufficiency of off-diagonal blocks}\label{sec:off-diag}
Analysis of $A(x^{(1)},\ldots,x^{(r)})$ is complicated by the fact
that all terms with repeated indices are zero.  Off-diagonal blocks of
$A$ are easier to analyze because no such terms exist.  Thankfully, as
Frieze and Kannan \cite{FK08} have shown, analyzing these off-diagonal
blocks suffices.  Here we generalize their proof to $r > 3$.

For a collection $\{V_1,V_2,\ldots,V_r\}$ of subsets of
$[n]$, we define
\[
A|_{V_1 \times \ldots \times V_r} (x^{(1)},\ldots,x^{(r)}) =
\sum_{k_1 \in V_1, \ldots, k_r \in V_r}
A_{k_1\ldots k_r} x^{(1)}_{i_1} x^{(2)}_{i_2} \ldots x^{(r)}_{i_r}
\]

\begin{lemma}\label{thrm:off-diag}
Let $P$ be the class of partitions of $[n]$ into $r$ equally sized
sets $V_1,\ldots, V_r$ (assume wlog that $r$ divides $n$).  Let $V =
V_1 \times \ldots \times V_r$.  Let A be a random tensor over $[n]^r$
where each entry is in $[-1,1]$ and let $R \subseteq B(0,1)$.  If for
every fixed $(V_1,\ldots V_r) \in P$, it holds that
\[
\Pr[ \max_{x^{(1)},\ldots,x^{(r)} \in R} A|_V (x^{(1)},\ldots,x^{(r)}) \geq f(n) ] \leq \delta,
\]
then
\[
\Pr[ \max_{x^{(1)},\ldots,x^{(r)} \in R} A(x^{(1)},\ldots,x^{(r)}) \geq 2r^r f(n) ] \leq \frac{\delta n^{r/2}}{f(n)},
\]
\end{lemma}
\begin{proof}[Proof of Lemma \ref{thrm:off-diag}]
Each $r$-tuple appears in an equal number of partitions and this
number is slightly more than a $r^{-r}$ fraction of the total. Therefore,
\begin{eqnarray*}
\left| A(x^{(1)},\ldots A(x^{(r)}) \right| & \leq & \frac{r^r}{|P|}
\left| \sum_{\{V_1,\ldots,V_r\} \in P} A|_V (x^{(1)},\ldots A(x^{(r)})\right|\\
& \leq & \frac{r^r}{|P|}
\sum_{\{V_1,\ldots,V_r\} \in P} \left| A|_V (x^{(1)},\ldots A(x^{(r)})\right|
\end{eqnarray*}
We say that a partition $\{V_1,\ldots,V_r\}$ is good if
\[
\max_{x^{(1)},\ldots,x^{(r)} \in R} A|_V (x^{(1)},\ldots,x^{(r)})
< f(n).
\]
Let the good partitions be denoted by $G$ and let $\bar{G} = P
\setminus G$.  Although the $f$ upper bound does not hold for
partitions in $\bar{G}$, the trivial upper bound of $n^{r/2}$ does
(recall that every entry in the tensor is in the range $[-1,1]$ and $R
\subseteq B(0,1)$).  Therefore
\[
\left| A(x^{(1)},\ldots A(x^{(r)}) \right| \leq r^r ( \frac{|G|}{|P|}
f + \frac{|\bar{G}|}{|P|} n^{r/2}).
\]
Since $E[|G|/|P|] = \delta$ by hypothesis, Markov's inequality gives
\[
\Pr[ \frac{|G|}{|P|} n^{r/2} > f] \leq \frac{\delta n^{r/2}}{f}
\]
and thus proves the result.
\end{proof}

\subsection{A concentration bound}\label{sec:concentration}
The following concentration bound is a key tool in our proof of
Theorem \ref{thrm:r-tensor-norm}.  We apply it for $t = \tilde{O}(N)$.
\begin{lemma}\label{thrm:concentration}
Let $\{u^{(i)}\}_{i=1}^{N}$ and $\{v^{(i)}\}_{i=1}^{N}$ be
collections of vectors of dimension $N'$ where each entry of $u^{(i)}$ is
$1$ or $-1$ with probability $1/2$ and $\|v^{(i)}\|_2 \leq 1$.  Then
for any $t \ge 1$,
\[
\Pr[ \sum_{i=1}^N (u^{(i)} \cdot v^{(i)})^2 \geq t]
\leq e^{-t/18} (4\sqrt{e\pi})^N.
\]
\end{lemma}
Before giving the proof, we note that this lemma is stronger than what
a naive application of standard theorems would yield for $t =
\tilde{O}(N)$.  For instance, one might treat each $(u^{(i)} \cdot
v^{(i)})^2$ as an independent random variable and apply a Hoeffding
bound.  The quantity $(u^{(i)} \cdot v^{(i)})^2$ can vary by as much
as $N'$, however, so the bound would be roughly $\exp(-ct^2/N{N'}^2)$
for some constant $c$.  Similarly, treating each $u^{(i)}_j$ as an
independent random variable and applying McDiarmid's inequality, we
find that every $u^{(i)}_j$ can affect the sum by as much as $1$
(simultaneously).  For instance suppose that every $v^{(i)}_j =
1/\sqrt{N'}$ and every $u^{(i)}_j = 1$.  Then flipping $u^{(i)}_j$
would have an effect of $|N' - ((N'-2)/\sqrt{N'})^2| \approx 4$, so
the bound would be roughly $\exp(-ct^2/NN')$ for some constant $c$.

\begin{proof}[Proof of Lemma \ref{thrm:concentration}]
Observe that
\(
\sqrt{\sum_{i=1}^N (u^{(i)} \cdot v^{(i)})^2}
\)
is the length of the vector whose $i$th coordinate is $u^{(i)} \cdot v^{(i)}$.
Therefore, this is also equivalent to the maximum projection of this vector onto a unit vector:
\[
\sqrt{\sum_{i=1}^N (u^{(i)} \cdot v^{(i)})^2} = \max_{y \in
 B(0,1)} \sum_{i=1}^N \sum_{j=1}^{N'} y_i u^{(i)}_j v^{(i)}_j.
\]

We will use an $\epsilon$-net to approximate the unit ball and give an
upper bound for this quantity.  Let $\mathcal{L}$ be the lattice
$\left(\frac{1}{2\sqrt{N}} \mathbb{Z}\right)^N$.
\begin{claim}\label{thrm:e-net-approx}
For any vector $x$,
\[
\|x\|_2 \leq 2 \max_{y \in \mathcal{L} \cap B(0,3/2)} y \cdot x.
\]
\end{claim}

Thus,
\[
\sqrt{\sum_{i=1}^N (u^{(i)} \cdot v^{(i)})^2} \leq
2 \max_{y \in \mathcal{L} \cap B(0,3/2)}
 \sum_{i=1}^N y_i \sum_{j=1}^{N'}  u^{(i)}_j v^{(i)}_j.
\]

Consider a \emph{fixed} $y \in \mathcal{L} \cap B(0,3/2)$.  Each
$u^{(j)}_i$ is $1$ or $-1$ with equal probability, so the expectation
for each term is zero.  The difference between the upper and lower
bounds for a term is
\[
2 | 2 y_j u^{(i)}_j v{(i)}_j | = 4 |y_j v{(i)}_j|
\]
Therefore,
\[
16 \sum_{i=1}^N \sum_{j=1}^{N'} (y_i u^{(i)}_j v{(i)}_j)^2
\leq 16 \sum_{i=1}^N y^2 \sum_{j=1}^{N'} (v{(i)}_j)^2 = 36.
\]

Applying the Hoeffding bound gives that
\[
\Pr[ \sum_{i=1}^N (u^{(i)} \cdot v^{(i)})^2 \geq t] \leq 
\Pr[ 2\sum_{i=1}^N y_i \sum_{j=1}^{N'} u^{(i)}_j v{(i)}_j \geq \sqrt{t} ]
\leq e^{-t/18}.
\]
The result follows by taking a union bound over $\mathcal{L} \cap B(0,3/2)$, whose cardinality is bounded according to Claim \ref{thrm:e-net-size}.
\end{proof}

\begin{claim}\label{thrm:e-net-size}
The number of lattice points in $\mathcal{L} \cap B(0,3/2)$ is at most $(4\sqrt{e\pi})^N$
\end{claim}
\begin{proof}[Proof of Claim \ref{thrm:e-net-size}]
Consider the set of hypercubes where each cube is centered on a
distinct point in $\mathcal{L} \cap B(0,3/2)$ and each has side length
of $(2\sqrt{n})^{-1}$.  These cubes are disjoint and their union
contains the ball $B(0,3/2)$.  Their union is also contained in the
ball $B(0,2)$.  Thus,
\begin{eqnarray*}
| \mathcal{L} \cap B(0,3/2) | & \leq & \frac{\Vol (B(0,2))}{(2\sqrt{N})^{-N}}\\
& \leq & \frac{\pi^{N/2} 2^N}{\Gamma(N/2 + 1)} 2^N N^{N/2}\\
& \leq & (4\sqrt{e\pi})^N.
\end{eqnarray*}
\end{proof}

\begin{proof}[Proof of Claim \ref{thrm:e-net-approx}]
Without loss of generality, we assume that $x$ is a unit vector.  Let
$y$ be the closest point to $x$ in the lattice.  In each coordinate
$i$, we have $|x_i - y_i| \leq (4\sqrt{n})^{-1}$, so overall $\|x -
y\| \leq 1/4$.

Letting $\theta$ be the angle between $x$ and $y$, we have
\[
\frac{x \cdot y}{\|x\| \|y\|}  = \cos \theta = \sqrt{1 - \sin^2 \theta}
\geq \left(1 - \frac{\|x - y\|^2}{\max\{\|x^2\|,\|y\|^2\}}\right)^{1/2}
\geq \sqrt{\frac{15}{16}}.
\]
Therefore,
\[
x \cdot y \geq \|y\| \sqrt{\frac{15}{16}}  \geq \frac{3}{4} \sqrt{\frac{15}{16}} \geq \frac{1}{2}.
\]
\end{proof}

\section{A bound on the norm of the parity tensor}\label{sec:r-tensor-norm}
In this section, we prove Theorem \ref{thrm:r-tensor-norm}.  First,
however, we consider the somewhat more transparent case of $r = 3$
using the same proof technique.
\subsection{Warm-up: third order tensors}
For $r=3$ the tensor $A$ is defined as follows:
\[
A_{k_1k_2k_3} = E_{k_1k_2} E_{k_2k_3} E_{k_1k_3}.
\]

\begin{theorem}
There is a constant $C_1$ such that with probability $1 - n^{-1}$
\[
\|A\| \leq C_1 \sqrt{n} \log^4 n.
\]
\end{theorem}
\begin{proof}
Let $V_1,V_2,V_3$ be a partition of the $n$ vertices and let $V = V_1 \times V_2 \times V_3$.  The bulk of the proof consists of the following lemma.
\begin{lemma}\label{thrm:main-3}
There is some constant $C_3$ such that
\[
\max_{x^{(1)}, x^{(2)}, x^{(3)} \in U}
A|_V(x^{(1)},x^{(2)},x^{(3)}) \leq C_3 \sqrt{n} \log n
\]
with probability $1 - n^{-7}$.
\end{lemma}

If this bound holds, then Lemma \ref{thrm:U-approx} then implies that
there is some $C_2$ such that
\[
\max_{x^{(1)}, x^{(2)}, x^{(3)} \in B(0,1)} A|_V(x^{(1)},x^{(2)},x^{(3)}) \leq C_2 \sqrt{n} \log^4 n.
\]
And finally, Lemma \ref{thrm:off-diag} implies that for some constant $C_1$
\[
\max_{x^{(1)}, x^{(2)}, x^{(3)} \in B(0,1)} A(x^{(1)},x^{(2)},x^{(3)}) \leq
C_1 \sqrt{n} \log^4 n
\]
with probability $1 - n^{-1}$ for some constant $C_1$.
\end{proof}

\begin{proof}[Proof of Lemma \ref{thrm:main-3}]
Define 
\begin{equation}\label{eqn:U_k}
U_k = \{ x \in U : |\Supp(x)| = k\}
\end{equation}
and consider a fixed $n \geq n_1 \geq n_2 \geq n_3 \geq 1$.
We will show that 
\[
\max_{(x^{(1)}, x^{(2)}, x^{(3)}) \in U_{n_1} \times U_{n_2} \times U_{n_3}}
A|_V(x^{(1)},x^{(2)},x^{(3)}) \leq C_3 \sqrt{n} \log n 
\]
with probability $n^{-10}$ for some constant $C_3$.  Taking a union bound over the $n^3$ choices of $n_1,n_2,n_3$ then proves the lemma.

We bound the cubic form as
\begin{eqnarray*}
\lefteqn{\max_{(x^{(1)}, x^{(2)}, x^{(3)}) \in U_{n_1} \times U_{n_2} \times U_{n_3}} A|_V(x^{(1)},x^{(2)},x^{(3)})}\\
 & = &
\max_{(x^{(1)}, x^{(2)}, x^{(3)}) \in U_{n_1} \times U_{n_2} \times U_{n_3}}
\sum_{k_1 \in V_1, k_2 \in V_2, k_3 \in V_3} A_{k_1k_2k_3} x^{(1)}_{k_1}
x^{(2)}_{k_2} x^{(3)}_{k_3}\\
& \leq & \max_{(x^{(2)}, x^{(3)}) \in U_{n_2} \times U_{n_3}}
\sqrt{\sum_{k_1 \in V_1} \left( \sum_{k_2 \in V_2, k_3 \in V_3}
A_{k_1k_2k_3} x^{(2)}_{k_2} x^{(3)}_{k_3}\right)^2}\\
& = &
\max_{(x^{(2)}, x^{(3)}) \in U_{n_2} \times U_{n_3}}
\sqrt{\sum_{k_1 \in V_1} \left( \sum_{k_2 \in V_2} E_{k_1k_2} x^{(2)}_{k_2}
\sum_{k_3 \in V_3} E_{k_2k_3} x^{(3)}_{k_3} E_{k_1k_3}\right)^2}.
\end{eqnarray*}
Note that each of the inner sums (over $k_2$ and $k_3$) are the dot
product of a random $-1,1$ vector (the $E_{k_1k_2}$ and $E_{k_2k_3}$
terms) and another vector.  Our strategy will be to bound the norm of
this other vector and apply Lemma \ref{thrm:concentration}.

In more detail, we view the expression inside the square root a
\begin{equation}\label{XX}
\sum_{k_1\in V_1}\left(
\underbrace{
\sum_{k_2 \in V_2}
\overbrace{E_{k_1k_2}}^{\mathstrut u^{(k_1)}_{k_2}}
\overbrace{
x^{(2)}_{k_2}
\underbrace{
\sum_{k_3 \in V_3}
\overbrace{E_{k_2k_3}}^{u^{(k_2)}_{k_3}}
\overbrace{x^{(3)}_{k_3} E_{k_1k_3}}^{\mathstrut v^{(k_1k_2)}(x^{(3)})_{k_3}}}_{u^{(k_2)} \cdot v^{(k_1k_2)}(x^{(3)})}}^{v^{(k_1)}(x^{(2)},x^{(3)})_{k_2}}}_{u^{(k_1)} \cdot v^{(k_1)}(x^{(2)},x^{(3)})}
\right)^2
\end{equation}

where $u^{(k_2)}_{k_3} = E_{k_2k_3}$ and $u^{(k_1)}_{k_2} =
E_{k_1k_2}$, while
\[
v^{(k_1k_2)}(x^{(3)})_{k_3} = x^{(3)}_{k_3} E_{k_1k_3}
\]
and
\[
v^{(k_1)}(x^{(2)},x^{(3)})_{k_2} = x^{(2)}_{k_2}
(u^{(k_2)} \cdot v^{(k_1k_2)}(x^{(3)})).
\]
Clearly, the $u$'s play the role of the random vectors and we will bound
the norms of the $v$'s in the application of Lemma \ref{thrm:concentration}.

To apply Lemma \ref{thrm:concentration} with $k_1$ being the index $i$, $u^{k_1}_{k_2} = E_{k_1k_2}$ above, we need a bound for every $k_1 \in V_1$ on
the norm of $v^{(k_1)}(x^{(2)},x^{(3)})$.  We argue
\begin{eqnarray*}
\lefteqn{\sum_{k_2} \left(x^{(2)}_{k_2}\sum_{k_3 \in V_3} E_{k_2k_3} x^{(3)}_{k_3} E_{k_1k_3}\right)^2}\\
& \leq &
\max_{k_1 \in V_1} \max_{x^{(2)} \in U_{n_2}}
\max_{x^{(3)} \in U_{n_3}}
\frac{1}{n_2} \sum_{k_2 \in \Supp(x^{(x_2)}}
\left(\sum_{k_3} E_{k_2k_3} x^{(3)}_{k_3} E_{k_1k_3}\right)^2\\
& = & F^2_1
\end{eqnarray*}
Here we used the fact that $\|x^{(2)}\|_\infty \leq n_2^{-1/2}$.
Note that $F_1$ is a function of the random variables $\{E_{ij}\}$ only.

To bound $F_1$, we observe that we can apply Lemma
\ref{thrm:concentration} to the expression being maximized above,
i.e.,
\[
\sum_{k_2}
\left(\sum_{k_3} E_{k_2k_3} \left(x^{(3)}_{k_3} E_{k_1k_3}\right)\right)^2
\]
over the index $k_2$, with $u^{k_2}_{k_3} = E_{k_2k_3}$. Now we need a
bound, for every $k_2$ and $k_1$ on the norm of the vector
$v^{(k_1k_2)}(x^{(3)})$.  We argue

\begin{eqnarray*}
\sum_{k_3} \left(x^{(3)}_{k_3} E_{k_1k_3}\right)^2
&\leq& ||x^{(3)}||^2_\infty \sum_{k_3} E_{k_1k_3}^2\\
&\leq& 1.
\end{eqnarray*}

Applying Lemma \ref{thrm:concentration} for a fixed $k_1, x^{(2)}$ and $x^{(3)}$ implies
\[
\frac{1}{n_2} \sum_{k_2 \in \Supp(x^{(2)})}
\left(\sum_{k_3} E_{k_2k_3} x^{(3)}_{k_3} E_{k_1k_3}\right)^2 > C_3 \log n
\]
with probability at most
\[
\exp(- \frac{C_3 n_2 \log n}{18}) (4\sqrt{e\pi})^{n_2}.
\]
Taking a union bound over the $|V_1| \leq n$ choices of $k_1$, and the
at most $n^{n_2}n^{n_3}$ choices for $x^{(2)}$ and $x^{(3)}$, we show
that
\[
\Pr[F_1^2 > C_3 \log n] \leq \exp( - \frac{C_3 n_2 \log n}{18}) (4\sqrt{e\pi})^{n_2} n  n^{n_2}
 n^{n_3}.
\]
This probability is at most $n^{-10}/2$ for a large enough constant $C_3$.

Thus, for a fixed $x^{(2)}$ and $x^{(3)}$, we can apply Lemma
\ref{thrm:concentration} to Eqn. \ref{XX} with $F_1^2 = C_3 \log n$ to
get:
\[
\sum_{k_1 \in V_1} \left(\sum_{k_2 \in V_2} E_{k_1k_2} \left(x^{(2)}_{k_2}
\sum_{k_3 \in V_3} E_{k_2k_3} x^{(3)}_{k_3} E_{k_1k_3}\right)\right)^2 > F_1^2 C_3 n \log n
\]
with probability at most $\exp(-C_3 n \log
n/18)(4\sqrt{e\pi})^{n}$. Taking a union bound over the at most
$n^{n_2}n^{n_3}$ choices for $x^{(2)}$ and $x^{(3)}$, the bound holds with probability
\[
\exp(-C_3 n \log
n/18)(4\sqrt{e\pi})^{n} n^{n_2} n^{n_3} \leq n^{-10}/2
\]
for large enough constant $C_3$.

Thus, we can bound the squared norm:
\begin{eqnarray*}
&&\max_{(x^{(1)}, x^{(2)}, x^{(3)}) \in U_{n_1} \times U_{n_2} \times U_{n_3}}
A|_V(x^{(1)},x^{(2)},x^{(3)})^2\\
&\leq& \sum_{k_1 \in V_1} \left(\sum_{k_2 \in V_2} E_{k_1k_2} \left(x^{(2)}_{k_2}
\sum_{k_3 \in V_3} E_{k_2k_3} x^{(3)}_{k_3} E_{k_1k_3}\right)\right)^2\\
& \leq & C_3^2 n_1 \log^2 n
\end{eqnarray*}
with probability $1 - n^{-10}$.
\end{proof}

\subsection{Higher order tensors}
Let the random tensor $A$ be defined as follows.
\[
A_{k_1,\ldots,k_r} = \prod_{1 \leq i < j \leq r} E_{k_ik_j}
\]
where $E$ is an $n \times n$ matrix where each off-diagonal entry is
$-1$ or $1$ with probability $1/2$ and every diagonal entry is $1$.

For most of this section, we will consider only a single off-diagonal
cube of $A$.  That is, we index over $V_1 \times \ldots \times V_r$
where $V_i$ are an equal partition of $[n]$.  We denote this block by
$A|_V$.  When $k_i$ is used as an index, it is implied that $k_i \in
V_i$.

The bulk of the proof consists of the following lemma.
\begin{lemma}\label{thrm:main-r}
There is some constant $C_3$ such that
\[
\max_{x^{(1)}, \ldots x^{(r)} \in U}
A|_V(x^{(1)},\ldots,x^{(r)})^2 \leq n(C_3 r \log n)^{r-1}
\]
with probability $1 - n^{-9r}$.
\end{lemma}

The key idea is that Lemma \ref{thrm:concentration} can be applied
repeatedly to collections of $u$'s and $v$'s in a way analogous to
Eqn. \ref{XX}.  Each sum over $k_r,\ldots, k_2$ contributes a $C_3r \log
n$ factor and the final sum over $k_1$ contributes the factor of $n$.

If the bound holds, then Lemma \ref{thrm:U-approx} implies that there
is some $C_2$ such that
\[
\max_{x^{(1)}, x^{(2)}, x^{(3)} \in B(0,1)} A|_V(x^{(1)},x^{(2)},x^{(3)})^2 \leq C_2^r r^{2r + r-1} n \log^{2r + (r-1)} n.
\]
And finally, Lemma \ref{thrm:off-diag} implies that for some constant $C_1$
\begin{eqnarray*}
\max_{x^{(1)}, x^{(2)}, x^{(3)} \in B(0,1)} A(x^{(1)},x^{(2)},x^{(3)}) & \leq &
C_1^{r} r^{2r+2r+(r-1)} n \log^{2r+r-1} n\\
& = & C_1^{r} r^{5r-1} n \log^{3r-1} n.
\end{eqnarray*}
with probability $1 - n^{-1}$ for some constant $C_1$.

\begin{proof}[Proof of Lemma \ref{thrm:main-r}]
 We define the set $U_k$ as in Eqn. \ref{eqn:U_k}.  It suffices to
 show that the bound
\[
\max_{(x^{(1)}, \ldots x^{(r)}) \in U_{n_1} \times \ldots \times U_{n_r}}
A|_V(x^{(1)},\ldots,x^{(r)})^2 \leq n (C_3 r \log n)^{r-1}
\]
holds with probability $1 - n^{-10r}$ for some constant $C_3$, since
we may then take a union bound over the $n^r$ choices of $n \geq n_1
\geq \ldots \geq n_r \geq 1$.

For convenience of notation, we define a family of tensors as follows
\begin{equation}\label{eqn:B-def}
B^{(k_1,\ldots,k_\ell)}_{k_{\ell+1},\ldots,k_r} = \prod_{i,j : i, \ell < j} E_{k_ik_j}
\end{equation}
where the superscript indexes the family of tensors and the subscript
indexes the entries.  Note that for every $k_1,\ldots,k_r \in V_1
\times \ldots \times V_r$, we have $B^{(k_1,\ldots,k_r)} = 1$, since
the product is empty.

Note that the tensor $B^{(k_1,\ldots,k_\ell)}$ depends only a subset of
$E$.  In particular, any such tensor of order $r - \ell$  will depend only on
the blocks of $E$
\[
F_\ell = \{E|_{V_i \times V_j} : i,\ell< j\}.
\]
Clearly, $F_r = \emptyset$, $F_1$ contains all blocks, and $F_{\ell}
\setminus F_{\ell+1} = \{E|_{V_i \times V_{\ell+1}} : i \leq \ell\}$.

We bound the $r$th degree form as
\begin{eqnarray}
\lefteqn{\max_{x^{(1)}, \ldots, x^{(r)} \in U_{n_1} \times \ldots \times U_{n_r}} A|_V(x^{(1)},\ldots,x^{(r)})} \nonumber\\
& = &
\max_{x^{(1)}, \ldots, x^{(r)} \in U_{n_1} \times \ldots \times U_{n_r}}
\sum_{k_1 \in V_1} x^{(1)}_{k_1}
B^{(k_1)}(x^{(2)},\ldots x^{(r)})\nonumber\\
& \leq &
\max_{x^{(2)}, \ldots x^{(r)} \in U_{n_2} \times \ldots \times U_{n_r}}
\sqrt{\sum_{k_1 \in V_1} B^{(k_1)}(x^{(2)},\ldots x^{(r)})^2}. \label{eqn:ssB}
\end{eqnarray}

Observe that for a general $\ell$,
\begin{equation}
B^{(k_1,\ldots,k_{\ell})}(x^{(\ell+1)},\ldots,x^{(r)}) =
\sum_{k_{\ell+1} \in V_{\ell+1}} E_{k_{\ell}k_{\ell+1}}
v^{(k_1,\ldots,k_{\ell})}(x^{(\ell+1)},\ldots,x^{(r)})_{k_{\ell+1}},
\end{equation}
where
\begin{equation}\label{eqn:v-def}
v^{(k_1,\ldots,k_{\ell})}(x^{(\ell+1)},\ldots,x^{(r)})_{k_{\ell+1}} = x^{(\ell+1)}_{k_{\ell+1}}
B^{(k_1,\ldots,k_{\ell+1})}(x^{(\ell+2)},\ldots,x^{(r)}) \prod_{i < \ell} E_{k_ik_{\ell+1}}.
\end{equation}
It will be convenient to think of
$B^{(k_1,\ldots,k_{\ell})}(x^{(\ell+1)},\ldots,x^{(r)})$ as the dot
product of a random vector $u^{(k_\ell)}$, where
$u^{(k_\ell)}_{k_{\ell+1}} = E_{k_\ell k_{\ell+1}}$ and
$v^{(k_1,\ldots,k_{\ell})}(x^{(\ell+1)},\ldots,x^{(r)})_{k_{\ell+1}}$, so that
\begin{equation}\label{eqn:B-dot-prod}
B^{(k_1,\ldots,k_{\ell})}(x^{(\ell+1)},\ldots,x^{(r)}) = u^{(k_\ell)} \cdot
v^{(k_1,\ldots,k_{\ell})}(x^{(\ell+1)},\ldots,x^{(r)}).
\end{equation}

The sum over $k_1 \in V_1$ from Eqn. \ref{eqn:ssB} can therefore be
expanded as
\[
\sum_{k_1 \in V_1} B^{(k_1)}(x^{(2)},\ldots x^{(r)})^2
= \sum_{k_1 \in V_1}
\left( u^{(k_1)} \cdot v^{(k_1)}(x^{(2)},\ldots,x^{(r)})\right)^2.
\]

Our goal is to bound $\|v^{(k_1)}(x^{(2)},\ldots,x^{(r)})\|$ and apply
Lemma \ref{thrm:concentration}.  Notice that for general $\ell$
\begin{multline}\label{eqn:v-bound-r}
\left \|
v^{(k_1,\ldots,k_{\ell})}(x^{(\ell+1)},\ldots,x^{(r)})\right\|_2^2\\
= \frac{1}{n^{\ell+1}}
\sum_{k_{\ell+1} \in \Supp(x^{(\ell+1)})}
B^{(k_1,\ldots,k_{\ell+1})}(x^{(\ell+2)},\ldots,x^{(r)})^2\\
\leq \max_{k_1,\ldots,k_\ell}
\max_{x^{(\ell+1)} \in U_{n_{\ell+1}}\ldots x^{(r)} \in U_{n_r}}\\
\frac{1}{n_{\ell+1}}
\sum_{k_{\ell+1} \in \Supp(x^{(\ell+1})})
B^{(k_1,\ldots,k_{\ell+1})}(x^{(\ell+2)},\ldots,x^{(r)})^2 = f_\ell^2.
\end{multline}
Note that the quantity $f_\ell$ (define above) depends only on the
blocks $F_{\ell+1}$.

The following claims will establish a probabilistic bound on $f_1$.
\begin{claim}
The quantity
\[
f_{r-1} = 1.
\]
\end{claim}
\begin{proof}
Trivially, every $B^{(k_1,\ldots,k_r)}()^2 = 1$.  Therefore, for every
subset $S_r \subseteq V_r$ such that $|S_r| = n_r$
\[
\frac{1}{n_r} \sum_{k_{r} \in S_{r}} B^{(k_1,\ldots,k_r)}()^2 = 1.
\]
\end{proof}
\begin{claim}\label{thrm:f-induct}
There is a constant $C_3$ such that for any $\ell \in 1 \ldots r-2$
\[
\Pr[ f_\ell^2 > C_3 r f_{\ell+1}^2 \log n] \leq n^{-12r}.
\]
\end{claim}

We postpone the proof of Claim \ref{thrm:f-induct} and argue that by
induction we have that
\[
f_1^2 \leq (C_3 r \log n)^{r-2}
\]
with probability $1 - n^{-12r}r \geq 1 - n^{-11r}$.

Assuming that this bound holds,
\[
v^{(k_1)}(x^{(2)},\ldots,x^{(r)}) \leq (C_3 r \log n)^{r-2}
\]
for all $k_1 \in V_1$ and $x^{(2)} \ldots, x^{(r)}$.
By Lemma \ref{thrm:concentration} then
\begin{eqnarray*}
\sum_{k_1 \in V_1} B^{(k_1)}(x^{(2)},\ldots x^{(r)})^2 & = &
\sum_{k_1 \in V_1}
\left(u^{(k_\ell)} \cdot
v^{(k_1)}(x^{(3)},\ldots,x^{(r)})\right)^2\\
& > & n (C_3 r \log n)^{r-1}
\end{eqnarray*}
with probability at most
\[
\exp \left( - \frac{C_3 r n \log n}{18} \right) (4\sqrt{e\pi})^{n}
\]
which is at most $n^{-11r}$ for a suitably large $C_3$.

Altogether the bound of the lemma holds with probability $1 -
2n^{-11r} \geq 1 - n^{-10r}$.
\end{proof}

\begin{proof}[Proof of Claim \ref{thrm:f-induct}]
Consider a fixed choice of the following: 1) $k_1,\ldots k_\ell$ and 
2) $x^{(\ell+1)} \in U_{n_{\ell+1}}, \ldots x^{(r)} \in
U_{n_r}$.  From Eqn. \ref{eqn:v-bound-r}, we have from definition that
for every $k_{\ell+1} \in V_{\ell+1}$
\[
\|v^{(k_1\ldots k_{\ell+1})}(x^{(\ell+2)},\ldots,x^{(r)})\|_2^2
\leq f_{\ell+1}^2.
\]
Therefore, by Lemma \ref{thrm:concentration}
\begin{eqnarray*}
\sum_{k_{\ell+1} \in \Supp(x^{(\ell+1)})}
B^{(k_1,\ldots,k_{\ell+1})}(x^{(\ell+2)},\ldots,x^{(r)})^2
& = & 
\sum_{k_{\ell+1} \in \Supp(x^{(\ell+1)})}
\left(u^{(\ell+1)} \cdot
v^{(k_1\ldots k_{\ell+1})}(x^{(\ell+2)},\ldots,x^{(r)})\right)^2 \\
& > & C_3 r f_{\ell+1}^2 n_{\ell+1} \log n
\end{eqnarray*}
with probability at most
\[
\exp\left(-\frac{C_3 r n_{\ell+1} \log n}{18}\right) (4\sqrt{e\pi})^{n_{\ell+1}}.
\]
Taking a union bound over the choice of $k_1,\ldots k_\ell$ (at most
$n^r$), and the choice of $x^{(\ell+1)} \in U_{n_{\ell+1}}, \ldots x^{(r)} \in
U_{n_r}$ (at most $n^{(r-1)n_{\ell+1}}$), the probability that
\[
f_\ell^2 > C_3 r f_{\ell+1}^2 \log n
\]
becomes at most
\[
\exp\left(-\frac{C_3 r n_{\ell+1} \log n}{18} \right) (4\sqrt{e\pi})^{n_{\ell+1}}
n^{r n_{\ell+1}}.
\]
For large enough $C_3$ this is at most $n^{-12r}$.
\end{proof}

\section{Finding planted cliques}\label{sec:reduction}
We now turn to Theorem \ref{thrm:reduction} and to the problem of
finding a planted clique in a random graph.  A random graph with a
planted clique is constructed by taking a random graph and then adding
every and edge between vertices in some subset $P$ to form the planted
clique.  We denote this graph as $G_{n,1/2} \cup K_p$.  Letting $A$ be
the $r$th order subgraph parity tensor, we show that a vector $x \in
B(0,1)$ that approximates the maximum of $A(\cdot)$ over the unit ball
can be used to reveal the clique, using a modification of the
algorithm proposed by Frieze and Kannan \cite{FK08}.

This implies an interesting connection between the tensor problem and
the planted clique problem.  For symmetric second order tensors
(i.e. matrices), maximizing $A(\cdot)$ is equivalent to finding the
top eigenvector and can be done in polynomial time.  For higher order
tensors, however, the complexity of maximizing this function is open
if elements with repeated indices are zero.  For random tensors, the
hardness is also open.  Given the reduction presented in this section,
a hardness result for the planted clique problem would imply a similar
hardness result for the tensor problem.

Given an $x$ that approximates the maximum of $A(\cdot)$ over the unit
ball, the algorithm for finding the planted clique is given in
Alg. \ref{alg:alg}.  The key ideas of using the top eigenvector of
subgraph and of randomly choosing a set of vertices to ``seed'' the
clique (steps 2a-2d) come from Frieze-Kannan \cite{FK08}.  The major
difference in the algorithms is the use of the indicator decomposition.
Frieze and Kannan sort the indices so that $x_1 \geq \ldots x_n$ and
select one set $S$ of the form $S = [j]$ where $\|A|_{S \times S}\|$
exceeds some threshold.  They run steps (2a-2d) only on this set.  By
contrast Alg. \ref{alg:alg} runs these steps on every $S =
\Supp(y^{(j)}(x))$ where $j = -\lceil r \log n \rceil,\ldots \lceil r
\log n \rceil$.

\begin{algorithm}[t]\label{alg:alg}
\caption{An Algorithm for Recovering the Clique}
\begin{tabbing}
Input:\\
1) Graph $G$.\\
2) Integer $p = |P|$.\\
3) Unit vector $x$.
\end{tabbing}
Output: A clique of size of $p$ or FAILURE.
\begin{enumerate}
\item Calculate $y^{-\lceil r \log n \rceil}(x), \ldots, y^{\lceil r \log n \rceil}(x)$ as defined in the indicator decomposition.
\item For each such $y^{(j)}(x)$, let $S = \Supp(y^{(j)}(x))$ and try the following:
\begin{enumerate}
\item Find $v$, the top eigenvector of the $1,-1$ adjacency matrix $A|_{S
\times S}$.
\item Order the vertices (coordinates) such that $v_1 \geq \ldots \geq
v_{|S|}$. (Assuming dot-prod is $\sqrt{1/2}$ below)
\item For $\ell = 1$ to $|S|$, repeat up to $n^{30} \log n$ times:
\begin{enumerate}
\item Select $10 \log n$ vertices $Q_1$ at random from $[\ell]$.
\item Find $Q_2$, the set of common neighbors of $Q_1$ in $G$.
\item If the set of vertices with degree at least $7p/8$, say $P'$ has
cardinality $p$ and forms a clique in $G$, then return $P'$.
\end{enumerate}
\item Return FAILURE.
\end{enumerate}
\end{enumerate}
\end{algorithm}

The algorithm succeeds with high probability when a subset $S$ is
found such that $|S \cap P| \geq C \sqrt{|S| \log n}$, where $C$ is an
appropriate constant.
\begin{lemma}[Frieze-Kannan]\label{thrm:recover}
There is a constant $C_5$ such that if $S \subseteq [n]$ satisfies $|S
\cap P| \geq C_5 \sqrt{|S| \log n}$, then with high probability steps
a)-d) of Alg. \ref{alg:alg} find a set $P'$ equal to $P$.
\end{lemma}
To find such an subset $S$ from a vector $x$, Frieze and Kannan
require that $\sum_{i \in P} x_i \geq C \log n$.  Using the indicator
decomposition, as in the Alg \ref{alg:alg}, however, reduces this to
$\sum_{i \in P} x_i \geq C \sqrt{\log n}$.  Even more importantly,
using the indicator decomposition means that only one element of the
decomposition needs to point in the direction of the clique.  The
vector $x$ could point in a very different direction and the algorithm
would still succeed.  We exploit this fact in our proof of Theorem
\ref{thrm:reduction}.  The relevant claim is the following.

\begin{lemma}\label{thrm:B-bound}
Let $B'$ be set of vectors $x \in B(0,1)$ such that
\[
|\Supp(y^{(j)}(x)) \cap P| < C_5 \sqrt{|\Supp(y^{(j)}(x))| \log n}
\] 
for every $j \in
\{-\lceil r \log n\rceil,\ldots, \lceil r \log n \rceil\}$.  Then, there is
a constant $C_1'$ such that with high probability
\[
\sup_{x \in B'} A(x,\ldots,x) \leq {C_1'}^r r^{5r/2} \sqrt{n}
\log^{3r/2} n.
\]
\end{lemma}
\begin{proof}
By the same argument used in the discretization, we have that
for any $x \in B'$
\begin{eqnarray}
A(x,\ldots,x) & \leq & (2\lceil r \log n\rceil)^r 
\max_{x^{(1)} \in Y^{(1)}(x), \ldots x^{(r)} \in Y^{(r)}(x)} 
A(x^{(1)},\ldots,x^{(r)}) \nonumber \\
& \leq &
(2\lceil r \log n\rceil)^r
\max_{x^{(1)}, \ldots x^{(r)} \in U'} A(x^{(1)},\ldots,x^{(r)}),
\label{eqn:disc-prime}
\end{eqnarray}
where
\[
U' = \{ |S|^{-1/2} \chi^S : S \subseteq [n], |S \cap P| < C_5 \sqrt{|S|
\log n}\}.
\]

Consider an off-diagonal block $V_1 \times \ldots \times V_r$.  For
each $i \in 1 \ldots r$, let $P_i = V_i \cap P$ and let $R_i = V_i
\setminus P$.  Then, breaking the polynomial $A|_V(\cdot)$ up as a sum
of $2^r$ terms, each corresponding to a choice of $S_1 \in \{P_1,R_1\},
\ldots, S_r \in \{P_r,R_r\}$ gives 
\begin{equation}\label{eqn:P-R-breakup}
\max_{x^{(1)},\ldots,x^{(r)} \in U'} A|_V(x^{(1)},\ldots,x^{(r)}) \leq
2^r
\max_{x^{(1)},\ldots,x^{(r)} \in U'} 
\sum_{S_1 \in \{P_1,R_1\},\ldots, S_r \in \{P_r,R_r\}}
A|_{S_1\times\ldots \times S_r}(x^{(1)},\ldots,x^{(r)}).
\end{equation}
By symmetry, without loss of generality we may consider the case where
$S_i = R_i$ for $i = 1 \ldots r-\ell$ and $S_i = P_i$ for $i=r-\ell+1
\ldots r$ for some $\ell$.  Let $\tilde{V} = R_1\times \ldots \times
R_{r - \ell} \times P_{r - \ell +1} \times \ldots \times P_{r}$.
Then,
\[
\max_{x^{(1)},\ldots,x^{(r)} \in U'} A|_{\tilde{V}}(x^{(1)},\ldots,x^{(r)}) =\\
\sum_{k_1 \in R_1} \ldots \sum_{k_{r - \ell} \in R_{r - \ell}}  
\prod_{i = 1 \ldots r-\ell}  x^{(i)}_{k_i}
\prod_{i,j:i,j \leq r - \ell} E_{k_ik_j} B^{(k_1,\ldots,k_{r-\ell})},
\]
where (as defined Eqn. \ref{eqn:B-def})
\[
B^{(k_1,\ldots,k_{r-\ell})}(x^{(r-\ell+1)},\ldots,x^{r})
\sum_{k_{r-\ell+1} \in P_{r-\ell+1}} \ldots \sum_{k_r \in P_r}  
\prod_{i = r-\ell+1 \ldots r}  x^{(i)}_{k_i}
\prod_{i,j:i,r-\ell+1 < j} E_{k_ik_j}.
\]
By the assumption that every $x^{(i)} \in U'$, this value is at most
$(C_5 \log n)^{\ell/2}$.  Thus,
\[
\max_{x^{(1)},\ldots,x^{(r)} \in U'} 
A|_{\tilde{V}}(x^{(1)},\ldots,x^{(r)}) \leq \\ 
\sum_{k_1 \in R_1} \ldots \sum_{k_{r - \ell} \in R_{r -
\ell}} \prod_{i = 1 \ldots r-\ell} x^{(i)}_{k_i} \prod_{i,j:i,j \leq r
- \ell} E_{k_ik_j} (C_5 \log n)^{\ell/2} .
\]
Note that every edge $E_{k_ik_j}$ above is random, so the polynomial
may be bounded according to Lemma \ref{thrm:main-r}.  Altogether,
\[
\max_{x^{(1)},\ldots,x^{(r)} \in U'}
A|_{\tilde{V}}(x^{(1)},\ldots,x^{(r)}) \leq (\max\{C_5,C_3\} \log
n)^{r/2}.
\]
Combining Eqn. \ref{eqn:disc-prime}, Eqn. \ref{eqn:P-R-breakup}, and
applying Lemma \ref{thrm:off-diag} completes the proof with $C_1'$
chosen large enough.
\end{proof}

\begin{proof}[Proof of Theorem \ref{thrm:reduction}]
The clique is found by finding a vector $x$ such that $A(x,\ldots,x)
\geq \alpha^r |P|^{r/2}$ and then running Algorithm \ref{alg:alg} on
this vector.  Algorithm \ref{alg:alg} clearly runs in polynomial time,
so the theorem holds if the algorithm succeeds with high probability.

By Lemma \ref{thrm:recover} the algorithm does succeed with high
probability when $x \notin B'$, i.e. when some $S \in
\{\Supp(y{-\lceil r \log n \rceil}(x), \ldots, \Supp(y{-\lceil r \log
n \rceil}(x)\}$ satisfies $|S \cap P| \geq C_5 \sqrt{|S| \log n}$.

We claim $x \notin B'$ with high probability.  Otherwise, for some $x \in B'$,
\[
A(x,\ldots,x) \geq \alpha^{r} p^{r/2} > 
C_0^r r^{5r/2} \sqrt{n} \log^{3r/2} n.
\]
This is a low probability event by Lemma \ref{thrm:B-bound} if $C_0
\geq C_1'$.
\end{proof}

\bibliographystyle{plain}
\bibliography{random}

\pagebreak

\appendix

\section{Proof of Lemma \ref{thrm:recover}}
Here, we give a Frieze and Kannan's proof of Lemma \ref{thrm:recover}
for the reader's convenience.  First, we show that the top eigenvector
of $A|_{S \times S}$ is close to the indicator vector for $S \cap P$.

\begin{claim}
There is a constant $C$ such that for every $S \subseteq [n]$ where
$|S \cap P| \geq C \sqrt{|S|\log n}$, the top eigenvector $v$ of the
matrix $A|_{S \times S}$ satisfies
\[
\sum_{i \in S \cap P} v_i > \sqrt{|S \cap P|/2}
\]
\end{claim}
\begin{proof}
The adjacency matrix $A$ can be written as the sum of $\chi^P
{\chi^P}^T$ and a matrix $R$ representing the randomly chosen edges.
Let $u = \chi^{S\cap P}/\sqrt{|S \cap P|}$ Suppose that $v$ is the top
eigenvector of $A|_{S \times S}$ and let $c = u \cdot v$.  Then
\begin{eqnarray*}
|S \cap P|^{1/2} & = & A(u,u)\\
& \leq & A|_{S \times S} (v,v)\\
& = & c^2 A|_{S \times S}(u,u) + 
2c\sqrt{1-c^2} A|_{S \times S}(u,v - cu) + 
(1 - c^2) A|_{S \times S}(v-cu, v - cu)\\
& \leq & c^2 |S \cap P|^{1/2} + 3\|R|_{S \times S}\|.
\end{eqnarray*}
Hence 
\[
c^2 \geq 1 - 3\frac{\|R|_{S \times S}\|}{C \sqrt{|S| \log n}}.
\]

By taking a union bound over the subsets $S$ of a fixed size, it
follows from well-known results on the norms of symmetric matrices
(\cite{FK81,Vu05}, also Lemma \ref{thrm:concentration}) that with high
probability
\[
\|R|_{S \times S}\| = O(\sqrt{|S| \log n})
\]
for every $S \subseteq [n]$.  Therefore, the theorem holds for a large
enough constant $C$.
\end{proof}

Next, we show that the clique is dense in the first $8|S \cap P|$
coordinates (ordered according to the top eigenvector $v$).

\begin{claim}\label{thrm:ell-dense}
Suppose $v_1 \geq \ldots \geq v_n$ and $\sum_{i \in S \cap P} v_i > \sqrt{|S
\cap P|/2}$.  Then for $\ell = 8|S \cap P|$
\[
|[\ell] \cap P| \geq \frac{|S \cap P|}{8}.
\]
\end{claim}

\begin{proof}[Proof of Claim \ref{thrm:ell-dense}]
For any integer $\ell$,
\begin{eqnarray*}
\sqrt{\ell} & \geq & \sum_{i \leq \ell} v_i \\
& \geq & \frac{\ell}{|S \cap P|} \sum_{i > \ell, i \in P} v_i\\
& = & \frac{\ell}{|S \cap P|} \left( \sum_{i \in P} v_i
- \sum_{i \leq \ell, i \in P} v_i \right)\\
& \geq & \frac{\ell}{|S \cap P|} \left( \sqrt{|S \cap P|/2} - \sqrt{|[\ell] \cap P|}\right).
\end{eqnarray*}
Thus,
\[
\sqrt{|[\ell] \cap P|} \geq \sqrt{|S \cap P|/2} - \frac{|S \cap
P|}{\sqrt{\ell}}.
\]
Taking $\ell = 8|S \cap P|$ (optimal), we have
\[
\sqrt{|[\ell] \cap P|} \geq \frac{1}{2\sqrt{2}} \sqrt{|S \cap P|}.
\]
\end{proof}

Given this density, it is possible to pick $10 \log n$ vertices from
the clique and use this as a seed to find the rest of the clique.
When $\ell = 8|S \cap P|$, in each iteration there is at least a
\[
8^{-10 \log n} = n^{-30}
\]
chance that $Q_1 \subseteq P$.  With high probability, no set of $10
\log n$ vertices in $P$ has more than $ 2 \log n$ common neighbors
outside of $P$ in $G$.  The contrary probability is
\[
\binom{|P|}{10 \log n} \binom{n}{2 \log n} 2^{-20 \log^2 n} = o(1).
\]
Letting $Q_2$ be the common neighbors of $Q_1$ in $G$, it follows that
$Q_2 \supseteq P$ and $|Q_2 \setminus P| \leq 2 \log n$.  Now, with
high probability no common neighbor has degree more than $3|P|/4$ in
$P$, because
\[
n \binom{|P|}{10 \log n} \binom{n}{2 \log n} \exp( - |P|/24) = o(1).
\]
for $|P| > 312 \log^2 n$.

Thus, with high probability no vertex outside of $P$ will have degree
greater than $7|P|/8$ in the subgraph induced by $Q_2$.
\end{document}